\documentclass[10pt,conference,compsocconf,letterpaper]{IEEEtran}

\usepackage{url}
\usepackage{cite}
\usepackage{graphicx}
\usepackage{amsmath}
\usepackage{latexsym}
\usepackage[boxruled]{algorithm2e}
\usepackage{algorithmic}
\usepackage{epsfig}
\graphicspath{{./figs/}}
\usepackage{amssymb,amsmath}

\newcommand{\note}[1]{{\underline{}}}
\newtheorem{theorem}{Theorem}[section]
\newtheorem{definition}[theorem]{Definition}
\newtheorem{lemma}[theorem]{Lemma}
\newtheorem{corollary}[theorem]{Corollary}

\title{Modeling DDoS Attacks by  Generalized Minimum Cut Problems}
\author{\authorblockN{Qi Duan, Haadi Jafarian and Ehab Al-Shaer\\}
\authorblockA{Department of Software and Information Systems\\
University of North Carolina at Charlotte \\
Charlotte, NC, USA\\}
\and
\authorblockN{Jinhui Xu\\}
\authorblockA{Department of Computer Science and Engineering\\
State University  of New York at Buffalo \\
Buffalo, NY, USA}
}

\begin{document}
\maketitle

\bibliographystyle{plain}

\begin{abstract}

Distributed Denial of Service (DDoS) attack is  one of the most preeminent threats in Internet. Despite considerable progress on this problem in recent years,
a remaining challenge is to determine its hardness by adopting proper mathematical models. In this paper, we propose to use  generalized minimum cut as basic tools to model
various types of DDoS attacks. Particularly, we study two important extensions of the classical minimum cut problem,
%In this paper, we study two important extensions of the classical minimum cut problem,
called {\em Connectivity Preserving Minimum Cut (CPMC)}
and {\em Threshold Minimum Cut (TMC)}, to model large-scale DDoS attacks.
% which have important applications in large-scale DDoS attacks.
In the CPMC problem,
a minimum cut is sought to separate a
  source node from a destination node and meanwhile preserve the connectivity
between the source and its partner node(s). The CPMC problem
also has important applications in many other areas such
as emergency responding, image processing,
pattern recognition, and medical sciences.  In the TMC problem, a minimum
cut is sought to isolate a target node from a threshold number of
partner nodes. TMC problem is an interesting special case of the network
inhibition problem and finds applications in network security.
 We  show that  the general CPMC problem  cannot be approximated within $logn$
 unless $NP=P$.
 % has quasi-polynomial algorithms.
 We also show that a special case of the CPMC problem
in planar graphs can be solved in polynomial time. The corollary of this
result is that the network diversion problem in planar graphs is in $P$; this settles
a previously open problem.
For the TMC problem,
we  show that the threshold minimum node cut (TMNC) problem  can be approximated within a ratio of $O(\sqrt{n})$
and the  threshold minimum edge cut   (TMEC) problem
can be approximated within a ratio of $O(\log^2{n})$.
As a consequence, we show that the related network inhibition problem and network interdiction problem cannot be approximated within any constant ratio
unless $NP \nsubseteq \cap_{\delta>0} BPTIME(2^{n^{\delta}})$. This settles another long standing open problem.
%
%\emph{This settles another long standing
%  open problem about the hardness of the
%network inhibition problem and the network interdiction problem.}
%We show that both problems
%cannot be approximated within any constant ratio
%unless $NP \nsubseteq \cap_{\delta>0} BPTIME(2^{n^{\delta}})$.
\end{abstract}

%\mbox{}
%\begin{keywords} Minimum Cut; Image Segmentation; Connectivity Preserving;\end{keywords}
%\begin{AMS}\end{AMS}

\section{Introduction}
\label{sec:intro}
%\subsection{Motivation}
Distributed Denial of Service (DDoS) attacks have become
one of most preeminent threats to Internet.
The DDoS attacks against
critical infrastructure (e.g., Internet backbone, power grid,
 financial services) are especially harmful.
For example, the Crossfire attack~\cite{cross} can disable up to
53\% of the Internet connections in some
US states, and up to about 33\% of  the connections in
the West Coast of the US.
Link/node flooding is an important form of DDoS attacks.
From the algorithmic point of view, link/node flooding
is closely related to the minimum cut problem. To better deal with such attacks, in this paper, we propose to use two generalized minimum cut problems to model them.

The basic minimum cut problem
 is one of the most fundamental problems in computer science and has numerous
applications in many different areas~\cite{Papadimitriou93,Vazirani04,PS98,Lawler01}.
In this paper, we investigate two important generalizations
of the minimum cut problem and their applications in link/node cut
based DDoS attacks.
 The first generalization is denoted as the {\em Connectivity Preserving Minimum Cut (CPMC)} problem,
which is to
find the minimum cut that separates a pair (or pairs) of source and destination
nodes and meanwhile preserve the connectivity between the source and its partner node(s).
The second generalization is denoted as the {\em Threshold Minimum Cut (TMC)} problem
in which a minimum
cut is sought to isolate a target node from a threshold number of
partner nodes.
The basic minimum cut problem tries to find a minimum node/edge
cut between a pair of nodes. If we want to find a minimum cut
to separate two  nodes (called {\em partner nodes}) from a third node, the minimum cut may
also separate the two partner nodes. For example,  in Fig~\ref{fig:cut_example}
we want to find a node cut to separate two partner nodes $s_1$ and $s_2$ from another
node $t$. If we choose node $s_5$ to be the cutting node,
then $s_1$ and $s_1$ will still be connected, which means
it is a connectivity preserving cut.
  If we choose node $s_4$ to be the cutting node,
then $s_1$ and $s_1$ will be disconnected, which means
it is not a connectivity preserving cut.  In many applications
it is desirable to find a connectivity preserving minimum cut
since it is a natural requirement to maintain some connectivity
when one wants to cut some links or nodes.

\begin{figure}[ht]
%\vspace{-0.2in}
\centering
\includegraphics[height=3.2in]{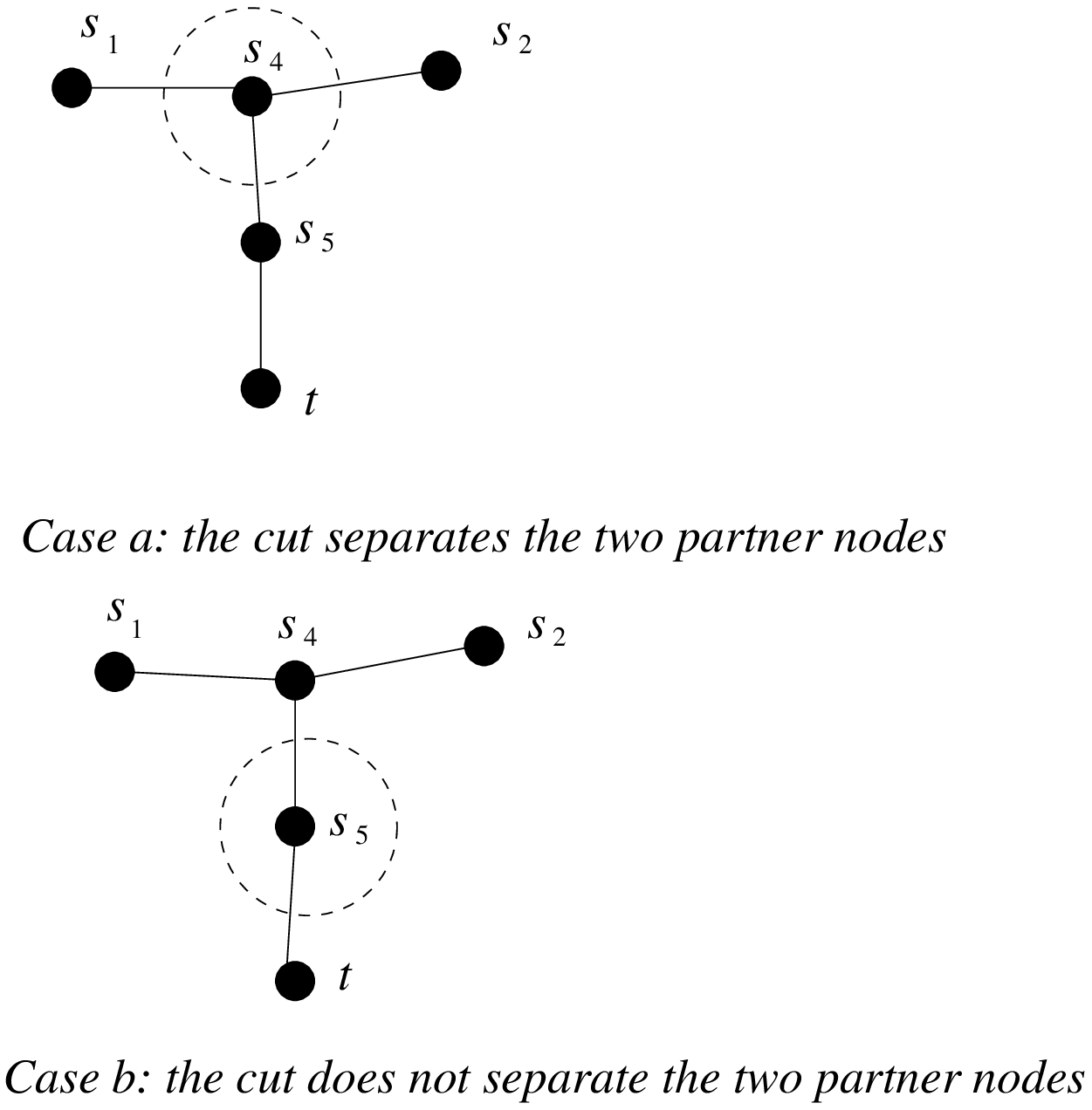}
\vspace{-0.3in}
\caption{An example of connectivity preserving minimum cut}
\label{fig:cut_example}
%\vspace{-0.25in}
\end{figure}

%888-549-3557 pc 4270604
The CPMC problem has  recently been studied in~\cite{cpmc14} (a
special case studied in~\cite{MCB14}). It
has  applications in many other areas, such as
 emergency responding,
data mining, pattern recognition, and machine learning.
It is very natural to have connectivity related constraints
in many minimum cut based applications since  ``cut''
is destructive and one may want to limit the destructive
nature of the minimum cut. In Botnet based DDoS attacks, it is important 
to find the cut to isolate the target server or area, while 
at the same time maintaining the connectivity between Bots
and Botmasters. In link cut attacks like Crossfire attack,
 it is important to avoid early congestion,
which is one kind of connectivity constraints.
In applications related to  emergency response, when a
gun attack such as the Sandy Hook elementary
school shooting~\cite{Sandy} happens in a building, the best response is to shut
down certain passages, and at the same time to
make sure that every one can have access to
some exits so that rescue personnels can reach them. This problem
is  closely related to CPMC.
In medical science, the  protein state transitions can be
modeled  as directed graphs called
biological regulatory networks (BRN)~\cite{BCRG04}. If one can identify the pathways that lead to
the states which cause cancer
or other possible diseases, then one should try
to find the best way to prevent the system from reaching
those dangerous states. For example, one may  find an optimal
 cut to disconnect the possible paths leading to those states, and at the same time maintain the
paths that are needed for normal metabolism. This is exactly the
CPMC problem.
The network diversion problem~\cite{Curet01, CWN13, CCDEGOQ01}
is a problem to consider the minimum cut to divert traffic to a certain link or set of links/nodes.
It has some similarity with the CPMC problem. However, the network diversion problem
defines the problem only in the context of  network diversion
(none of the works recognizes the more important CPMC problem), while the CPMC problem
 defines the problem in a more  natural way and has a lot more applications. In directed graphs,
the network diversion problem is not even an NP optimization problem since the
desired cut may not exist  and it is  NP-complete to judge if node disjoint paths between
 two pairs of nodes exist~\cite{Tholey04}. The CPMC problem is an NP optimization
problem in both undirected and directed graphs.
 Also, in planar graphs, the network diversion problem cannot be reduced to
the CPMC problem.
  Even without considering all the applications, the CPMC
 is a very natural problem in pure graph theory and discrete mathematics.

The TMC problem
arising naturally from DDoS  attacks, threshold cryptography,
 and distributed data storage, which concerns with blocking
a node from a threshold number of related nodes.
Threshold cryptography and threshold
related protocols  have wide applications
 in network security such as
secure and reliable cloud storage service~\cite{CYYLH12},
secure key generation and sharing
in mobile ad hoc networks~\cite{CAG04, CGA05, SKMO09}, etc.
 The natural optimization problem arises from
 threshold based protocols
is to block a node from a threshold number of related nodes
with minimum cost. This optimization problem
is important for both attackers and defenders.
From the attacker's point of view, he/she needs to
find an optimal way to thwart the execution of
the threshold based protocol, or crack some information
by compromising a threshold number of nodes or links. On the other hand, from
the defender's point of view, the defender may
need to block the communication between a
Bot master and a threshold number of Bots
to thwart the Botnet attacks.
In link cut DDoS attacks like  Crossfire
and Coremelt~\cite{core} attacks,
the attacker may try to achieve the desired
degradation ratio by flooding a  set of critical links
with minimum amount attacking flows. The problem
can be easily converted to TMC. In distributed cloud
storage, the attacker may try to disconnect
the user from a threshold number of cloud servers to
disrupt the services that require a certain number
of available servers, the problem is exactly TMC.
The TMC problem is an important special
case of the network inhibition
problem or network interdiction problem, and was first
studied in~\cite{DM13}.

The main results of this paper are the follows: (1)
We show that CPMC in directed graphs and multi-node
CPMEC are hard to approximate, 3-node planar CPMNC is in $P$,
and the network diversion problem in undirected planar graphs is in $P$, which
settles a long standing open problem~\cite{CWN13}.
(2) We show that the TMNC problem
can be approximated within a ratio of $\sqrt{n}$
and the TMEC problem
can be approximated within a ratio of $\log^2{n}$.
We reveal the relationship between
the  TMC  problem and some other closely
related problems. Particularly, we show that the
network inhibition problem~\cite{Phillips93} and the network interdiction problem~\cite{AEU10}
cannot be approximated within any constant ratio
unless $NP \nsubseteq \cap_{\delta>0} BPTIME(2^{n^{\delta}})$. This settles another long standing open problem.

The rest of the paper is organized as follows.
Section~\ref{sec:cpmc} investigates the hardness
of the CPMC problem. Section~\ref{sec:planar-cpmc}
discusses the hardness of CPMC in planar graphs.
Section~\ref{sec:thresh} and Section~\ref{sec:alg}
discusses the hardness and algorithms of the TMC
problem, respectively. Future work is presented
in Section~\ref{sec-fw}.
%
% also give an answer to another
%long standing  open problem: the hardness of the
%network inhibition problem~\cite{Phillips93} and the network interdiction problem~\cite{AEU10}.
%We show that both of them
%cannot be approximated within any constant ratio
%unless $NP \nsubseteq \cap_{\delta>0} BPTIME(2^{n^{\delta}})$.
%

\section{Hardness Results of CPMC}
\label{sec:cpmc}
We adopt the notation from~\cite{cpmc14}.
The most simple case of CPMC is the 3-node CPMC.
 Informally speaking, in the 3-node CPMC
problem, we are given a connected graph $G=(V,E)$ with positive node (or edge) weights,
a source node $s_1$ and its partner node $s_2$, and a destination node $t$. The objective
is to compute a cut with minimum weight to disconnect the source $s_1$ and destination $t$, and
meanwhile preserve the connectivity of $s_1$ and its partner node
 $s_2$ (i.e., $s_{1}$ and $s_{2}$ are connected after the cut). The weights
can be associated with either the nodes (i.e., vertices) or the edges,
and accordingly the cut can be either a set of nodes, called a connectivity
preserving node cut, or a set of edges, called a connectivity preserving edge cut.
In the former case, a cut is a subset of vertices $V$ whose removal (along with the edges
incident to them) disconnects $s_1$ and $t$, but does not affect the connectivity
 of $s_1$ and $s_2$. Such a cut is called a {\em connectivity preserving node
cut (CPMNC)}. In the latter case, a cut is a subset of edges whose removal disconnects $s_1$ and $t$ and preserves the connectivity of $s_1$ and $s_2$. Such a cut is called a  {\em connectivity preserving minimum edge cut (CPMEC)}.
The weight of a cut $C$ is the total weight associated with the nodes or edges in $C$.
Note that we can easily extend the 3-node CPMC problem
to the general case CPMC where
one may have multiple pairs of source and destination nodes, and each
source node may have multiple partner nodes.

First we note that the CPMNC problem is an NP optimization problem. To determine whether a valid cut exists, one just needs to check if $t$ is connected to any bridge node between $s_1$ and $s_2$; if so, then no valid cut exists. Clearly, this can be done in polynomial time. Thus, we assume thereafter that a cut always exists.

The decision version of the 3-node CPMNC problem is as follows:
given an undirected graph $G=(V,E)$ with each node $v_i \in V$ associated with a positive integer weight $c_i$, a source node $s_1$, a partner node $s_2$, a destination node $t$, and an integer $B>0$, determine whether there exists a subset of nodes in $V$ with total weight less than or equal to $B$ such that the removal of this subset disconnects $t$ from $s_1$ but preserves the connectivity between $s_1$ and $s_2$.

%\vspace{-0.15in}

The decision version of the 3-node CPMEC problem can be defined similarly:
given an undirected graph $G=(V,E)$ with each edge $e_i \in E$ associated with a positive
integer weight $c'_i$, a source node $s_1$, a partner node $s_2$, a destination
node $t$, and an integer $B'>0$, determine whether there exists a subset of edges in $E$
with total weight less than or equal to $B'$ such that the removal of this subset disconnects $t$ from $s_1$ but preserves the connectivity between $s_1$ and $s_2$.

The CPMEC
  has several
key differences from CPMNC. First,
the resulting graph after the node cut in CPMNC may be
disintegrated into many connected components,
where in CPMEC the resulting graph  has exactly
two connected components (otherwise there will
be some redundant edges, and the cut cannot be the
minimum one). Second, suppose the weight
of the  minimum edge cut between
a single node $s_1$ and destination $t$ is $C_e(s_1,t)$,
the weight of the minimum edge cut
(not necessarily connectivity preserving)
 between two nodes $s_1,s_2$ and destination $t$
 is $C_e(s_1,s_2,t)$, if $C_e(s_1,t)+C_e(s_2,t) > C_e(s_1,s_2,t)$,
then the  minimum edge cut  between two
nodes $s_1,s_2$ and destination $t$ must be
connectivity preserving.
Node cut does not
has this property. These key differences
mean that the node cut problem and the
edge cut problem may have different hardness.
In our NP-hardness proof of the CPMNC, we cannot modify it to get
a proof for the CPMEC.

Given nodes $s_1$ and $t$ in a graph, we can
classify other nodes into several categories.
If a node $s_2$ has the property
 $C_e(s_1,t)+C_e(s_2,t) > C_e(s_1,s_2,t)$,
  then it is easy to show that the minimum  edge cut  between two
nodes $s_1,s_2$ and destination $t$ must be
connectivity preserving.

\begin{lemma}
For two points $s_1$ and $s_2$ in the graph,
if $C_e(s_1,t)+C_e(s_2,t) > C_e(s_1,s_2,t)$,
then the minimum  edge cut  between two
nodes $s_1,s_2$ and destination $t$ must be
connectivity preserving.
\end{lemma}

\begin{proof}
For a minimum  edge cut  between two
nodes $s_1,s_2$ and destination $t$,
it must be the  the union
of two cuts: one is the
cut between $s_1$ and $t$, another
is  between $s_2$ and $t$.
The sum of this two cut
is at least  $C_e(s_1,t)+C_e(s_2,t)$.
If $C_e(s_1,t)+C_e(s_2,t) > C_e(s_1,s_2,t)$,
then the two cuts must have some
common edges. But if a common edge
exists,
 then
we have two cases:

  case 1: The common edge
is connected with the component
of $s_1$ and $s_2$, then this edge can be
removed from the cut, and the remaining
cut is still valid, and $s_1$ and $s_2$
is now connected.

case 2: The common edge
is connected with the two components
and the $t$ component, in this case,
$s_1$ and $s_2$ must be connected.

\end{proof}

If  $C_e(s_1,t)+C_e(s_2,t) = C_e(s_1,s_2,t) = C_{ep}(s_1,s_2,t)$
(here $C_{ep}(s_1,s_2,t)$ is the CPMEC between $s_1$, $s_2$, and $t$),
we call $s_2$ a threshold node of $s_1$.
time. If
node $s_2$ satisfies
 $C_{ep}(s_1,s_2,t) > C_e(s_1,t)+C_e(s_2,t)$, we
call them outer points of $s_1$.

We also investigate the multiple-partner CPMEC problem.
In this case, we have $u+1$ nodes $s_1, \ldots, s_{\tau},t$ in the graph, and
the objective is to find a minimum edge cut that
separates $ s_1, \ldots,  s_{\tau}$  from $t$, and at the same time
keeps   $ s_1, \ldots,  s_{\tau}$   connected.

Note that this problem is still an NP optimization
problem. If the removal of node $t$ causes some of the $s_i$ nodes to be
 disconnected from others, then no valid cut exists.
 Since this can be determined in polynomial time, we always assume
 that there exists a solution to the problem.

 In some applications, we need
to find a minimum cut to separate
two connected components $\Gamma_1$ and
$\Gamma_2$ from
another connected component $T$ in a graph,
and keep $\Gamma_1$ and $\Gamma_2$ connected.
This is a generalization of the
original 3-node connectivity preserving
minimum cut problem. We can show that
 the generalized problem has the same approximability
as the original problem, in both the
cases of node cut and edge cut.

\begin{theorem}
\label{thm:shr}
The generalized connectivity preserving
minimum cut problem can be L-reduced
to the original connectivity preserving
minimum cut problem. This means
the generalized problem has the same approximability
as the original connectivity preserving
minimum cut problem. In other words,
if the   connectivity preserving
minimum cut problem can be approximated
within $f(n)$ (where $n$ is the input size),
then the  generalized problem
can also be  approximated within $f(n)$.
\end{theorem}

\begin{proof}
To see this, we can
transform the connected components
$\Gamma_1$,$\Gamma_2$, and $T$ to three nodes. For
every connected  component, we can shrink the
whole component into one new node. All
edges inside the connected component
are deleted. For
every edge connecting a node in the component
and a node outside the component, we add
an edge between the new node and the outside node.
To make the graph still a simple graph,
if there are multiple edges between two nodes,
we add an intermediate node in the edge.
For node cut, we set the weight
of the intermediate node to be infinity. For
edge cut, the two intermediate edges all
have the same weight as the old edge.
An example of component shrinking is shown
in Fig. ~\ref{fig:shr}.

\begin{figure}
\centering
\includegraphics[width=3in]{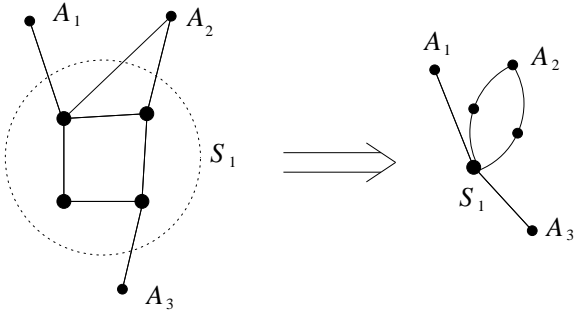}
\caption{Transformation of the Connected Component}
\label{fig:shr}
\end{figure}

 Now it is easy to see,
any  connectivity preserving
minimum cut in the old graph
  is the   connectivity preserving
minimum cut for the new graph, and
vice versa. Also note that
after the shrinking procedure,
the size of the graph decreased.
So the generalized problem
has the same approximability as the original
problem.

\end{proof}

To study the 3-node CPMEC problem in directed graphs, we first give its definition.
% of CPMC in directed graph.

\begin{definition} [ \bf 3-node CPMEC in directed graphs]
Let $G=(V,E)$ be a directed graph with $n$ nodes and $m$ edges. Each
edge $e_i \in E$ ($1\leq i\leq m$) is associated with a
 positive integer weight. Given three nodes $s_1$, $s_2$, $t$,
 and a positive integer $b$, the $3$-node CPMEC problem for $(s_{1}, s_{2}, t)$ is to
seek a subset $C$ of edges in $E$ with total weight less than or equal to $b$ such that
after the removal of $C$, there is no path  from $t$  to $s_1$
and $s_2$ and meanwhile  there exists at least one
path from $s_1$ and $s_2$.
\end{definition}

Below we show that the 3-node CPMEC problem on directed graphs
cannot be approximated within a logarithm ratio.

\begin{theorem}
\label{thm:di}
In a directed graph $G$, the 3-node connectivity preserving minimum edge cut
problem cannot be approximated within a factor of
$\alpha logn$ for some constant $\alpha$ unless
$P=NP$.
\end{theorem}

\begin{proof}
To prove the theorem, we reduce the set cover problem
to this problem. In the set cover  problem, we have a
ground set $\mathcal{T}=\{e_1, e_2, \ldots, e_{n_1} \}$ of $n_1$ elements, and
a set  $\mathcal{S}= \{S_1, S_2, \ldots, S_k \}$ of $k$ subsets of $\mathcal{T}$ with
each $S_i \in \mathcal{S}$ associated with a weight $w_i$. The objective is to select
a set $\mathcal{O}$ of subsets in $\mathcal{S}$ so that the union of all subsets in $\mathcal{O}$ contains every element in $\mathcal{T}$ and the total weight of subsets in $\mathcal{O}$ is minimized.

Given an instance $I$ of the set cover problem with $n_1$ elements and $k$ sets, we
construct a new graph. The new graph has an element gadget for every element,
 and every element gadget contains $k_1+2$
nodes, where $k_1$ is the number
of sets that contains this element.
 In every gadget, there are two end
points, and $k_1$ internal
nodes  are connected to the two end nodes
in parallel. Every internal  nodes
of a gadget corresponds to a set that
contains this element.
 All such $n_1$
gadgets are connected sequentially
through their end points,
with $s_1$ and $s_2$ at the two ends
of the whole construction.
We also construct an arc
for every set. There is an arc
from the ending
point of every set arc
to the corresponding set node in the element gadget,
and there is another  arc
from $t$
to the starting point of every set arc.

Figure ~\ref{fig:set_di}
is the graph constructed for set cover
instance with three elements $x_1$, $x_2$, and
$x_3$, three sets $A_1=\{ x_1,x_3\}$,  $A_2=\{ x_2,x_3\}$,
and   $A_3=\{ x_1,x_2\}$.

 \begin{figure}
\centering
\includegraphics[width=3in]{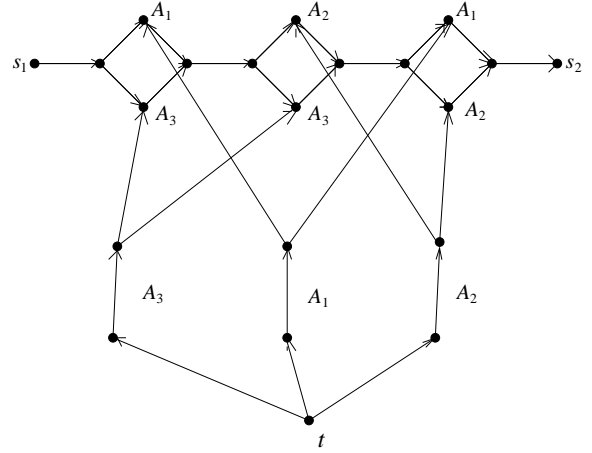}
\caption{An example illustrating the proof of Theorem \ref{thm:di}.}
\label{fig:set_di}
\end{figure}

\vspace{10pt}

 Every set arc is assigned with weight
$w_in_1k$, where $w_i$ is the weight of the
set in the original set cover instance.
All the arcs in a gadget are assigned  with
weight $1$. All other arcs ( the arcs
connecting the $s_1$, the gadgets, and $s_2$,
and the arcs connecting $t$ with the set arcs)
have weight infinity.
We also let $b=n_1kD_1 + n_1k-1$,
where $D_1$ is the bound of weight
in the set cover instance.
First note that we only need to cut $t$ from reaching $s_2$,
since there is no way for $t$ to reach $s_1$ in the
constructed graph.
Also note that one cannot
put all arcs from the set node to the
right end node of the gadget into the cut  in an element
gadget, otherwise $s_1$ and $s_2$
will be separated. Now we can see that if the set
cover instance has a cover with weight
no more $D_1$, then we can choose the
following cut:  The cut contains those
set arcs contained in the  cover and
all the gadget arcs starting with the
set node which are not
in the set cover. The cut has a weight
 $n_1kD_1 + g_1$,  where $g_1<n_1k$. Similarly
if we can find a cut with weight no more than
  $n_1kD_1 + n_1k-1$, then we can find a
set cover with weight no more than $D_1$.
Furthermore, since set cover cannot
be approximated within $\alpha logn$
for some constant $\alpha$ unless NP=P~\cite{RS97,Feige98},
we can see that the connectivity preserving minimum cut problem
in directed graph
cannot be approximated within  $\alpha_1logn$
for some constant $\alpha_1$ unless NP=P.
Suppose the optimal solution of the set
cover instance is $D$, then the optimal solution
of the constructed graph has a minimum cut
with weight $n_1kD+g_2$, where $0<g_2<n_1k$.
If we can find a cut
in which the total weight (in the set cover instance) of the set nodes
is $D_1$, then the cut has a weight  $n_1kD_1 + g_1$,
  where $0<g_1<n_1k$.
Assume   \[ \frac{n_1kD_1+g_2}{n_1kD + g_1} <\alpha_1 log(n_1k), \]
for some $\alpha_1$,
 then we have
 \[ \frac{D_1}{D} <\frac{n_1kD_1+g_2}{n_1kD + g_1}
+ o(1) < \alpha_1log(n_1k).\]

For the set cover problem with $n_1$ elements
and $k=poly(n_1)$ sets, it cannot be approximated
 within $\alpha logn_1$ unless
NP=P~\cite{RS97,Feige98}. Since $k$ is bounded by some polynomial
in $n_1$,
 we can see
   \[ \frac{D_1}{D}  < \alpha_1log(n_1k)
  \leq \alpha_1\alpha_2logn_1 ,\]
  where $\alpha_2$ is another constant. If we choose
$\alpha_1 \leq  \alpha/\alpha_2$, then
    \[ \frac{D_1}{D} \leq \alpha logn. \]
Now we have a contradiction, which means that the problem
cannot be approximated within
$\alpha logn$ unless
NP=P.
%\qed
 \end{proof}

Next we consider the case of multiple  partner nodes for the CPMC problem.
In this case, we have $u+1$ nodes $s_1, \ldots, s_u,t$ in the graph, and
the objective is to find a minimum edge cut that
separates $ s_1, \ldots,  s_u$  from $t$, and at the same time
keeps   $ s_1, \ldots,  s_u$   connected.

\iffalse
Note that this problem is still an NP optimization
problem. If the removal of node $t$ causes some of the $s_i$ nodes to be
 disconnected from others, then no valid cut exists.
 Since this can be determined in polynomial time, we always assume that there exists a solution to the problem.
\fi

\begin{theorem}
\label{the-ec2}
For an undirected graph $G$, the CPMC problem with $k$ partner nodes
(where $k$ is an integer, not necessarily constant, given in the input)
 cannot be approximated within a factor of $\alpha logn$ for some constant $\alpha $ unless
$NP=P$.
\end{theorem}

\begin{proof}
We use similar reductions as
in the proofs of Theorems \ref{thm:di}.
For the reduction of $\alpha logn$-inapproximability,
the only difference is that now the set node is replaced by a set edge,
as shown  in Figure ~\ref{fig:set_mul}.

\begin{figure}
\centering
\includegraphics[width=3in]{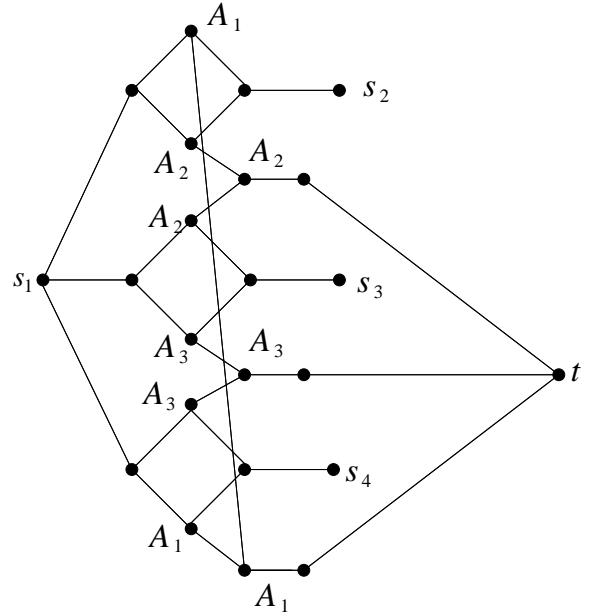}
\caption{An example illustrating the proof of Theorem \ref{the-ec2}.}
\label{fig:set_mul}
\end{figure}
As in the proof of Theorem \ref{thm:di},
the set edge has weight  $w_in_1k$.
Also in every  gadget,
the two edges connecting the end
nodes in the gadget and the set
node have weight 1, which are called gadget-set edges.
All other edges have weight infinite. Now it is easy to
see that to make $s_1$ disconnected from
$t$, for every set node  in a gadget, either the two gadget-set edges
 or the corresponding set edge must be included in the
cut. The remaining proof is the same as that in the proof of Theorem \ref{thm:di}.
\vspace{10pt}
\end{proof}

Similarly we can apply the proof of polylog factor
inapproximation  of CPMNC in~\cite{cpmc14}, using a similar
construction, to show that directed CPMC and multi-node CPMEC
cannot be approximated within a polylog factor. Due to space
limitation, it is omitted here.

\section{CPMC in Planar Graphs}
\label{sec:planar-cpmc}
 In~\cite{cpmc14} it is shown that 3-node CPMEC in planar graphs has polynomial
 time solutions and in~\cite{BL14} it is shown that
 multi-node CPMEC in planar graphs has polynomial
 time solutions.

 With some modifications,
the algorithm in~\cite{cpmc14} can also be applied
to 3-node planar CPMNC, and we can show
the algorithm also finds the optimal solution for  3-node planar CPMNC.

\begin{theorem}
The 3-node planar CPMNC can be solved in polynomial time.
\end{theorem}

\begin{proof}
The major difference between edge cut and node cut
is that in any node cut the graph may
be disintegrated into multiple connected components.
For any cut between a node $v$ and destination node $t$, we define
the connected component containing $v$ after the cut as the principal
cut component of node $v$.
We also introduce the perturbation
technique of node weight that we used in the proof of 3-node planar
CPMEC in~\cite{cpmc14}.
 In this way, any two cuts (or more generally,
two subsets of edges) in $G$ will
 have different weights unless they are completely identical. Based on this property, we
 have the following observations. (1) Any cut  is unique. (2) Given
 any node $v\in V$, let $C_{v}$ be the connected component containing $v$ and resulting
from the minimum edge cut between $v$ and $t$. Then all
%   minimum edge cut between
%  $v$ and $t$ cuts $V$ into two
%  connected components, the component containing $v$ is denoted
%  as $C_v$, then all
 nodes in $C_v$ can be uniquely determined due to the  perturbation technique.
Now note that any two principal cut components cannot enclose a hole.
Suppose the CPMNC principal cut component  between nodes $s_1,v$ and
$t$ is $\tilde{C}_{s_1,v}$. If there exists a hole that is completely surrounded by the two CPMNC
 principal cut components $\tilde{C}_{s_{1},A}$ and $\tilde{C}_{s_{1},A_{1}}$,
then for the nodes adjacent to the boundary of the hole, we can divide the nodes into
 3 types: The first type of  is in the CPMNC cutting
of $\tilde{C}_{s_{1},A}$ but not cutting nodes of  $\tilde{C}_{s_{1},A_{1}}$. We denote the total weight
of this type of nodes as $L$. The second type  is the cutting nodes
of $\tilde{C}_{s_{1},A_1}$ but not cutting nodes of  $\tilde{C}_{s_{1},A}$. We denote the total weight
of this type of nodes as $L_1$. The third type of segments is the cutting nodes of both
of $\tilde{C}_{s_{1},A_1}$  and  $\tilde{C}_{s_{1},A}$.
We have $L>L_1$,  $L_1>L$ or $L_1=L$. If $L>L_1$ then we remove the cutting nodes
of type 1 and add cutting nodes of type 2 for the CPMNC between $s_1,A$ and $t$.
Now we can get a smaller CPMNC, because
the new cut will decrease by a value of $L$ and increase  by a value of $L_1$,
and the overall effect is that the value of the cut will decrease by at least
$L-L_1$. This is a contradiction. Note that the principal
cut component $\tilde{C}_{s_{1},A}$ may not enlarge since the
added nodes in the hole may not connect to component $\tilde{C}_{s_{1},A}$, but
this does not affect the argument. The remaining two cases are similar.
So for planar CPMNC, any two principal cut components cannot enclose a hole.
The remaining proof of the 3-node planar CPMNC is similar to 3-node planar CPMEC
when we replace the enclosed edge cut region  $C_{s_{1},v}$ with principal cut component
$\tilde{C}_{s_{1},v}$ for any node $v$ which is not $t$.

\end{proof}

Next we show that 2-node versus 2-node planar CPMEC
is in $P$. Note that the problem is  an NP optimization
problem since the feasibility of disjoint paths for two source-destination
pairs in planar graphs can be determined in P~\cite{KS10}.

\begin{theorem}
Let $G=(V,E)$ be a planar graph,  and $s_{1}, s_{2}$ and $s'_{1}, s'_{2}$ be two pairs of  partner nodes in $G$.
%Suppose in a planar graph $G=(V,E)$, there are
%two group of partner nodes $s_1,   s_{2}$
%and $s'_1 , s'_{2}$.
Then  the CPMEC separating partner nodes $s_{1}, s_{2}$ from  partner nodes $s'_{1}, s'_{2}$  can be computed
in polynomial time.
\end{theorem}

\begin{proof}
First it is easy to see if $s_1$ and  $s_2$ are in the same face, then it is trivial, so
we ignore this case. Now suppose $s_1$ and $s_2$ are not in the
same face.
In the  algorithm for 3-node planar CPMEC, for every places when we
need to compute a minimum cut between a shrunk node
and the destination $t$ we can change it to ``compute the CPMEC between the shrunk node
and the two partner nodes $s'_1 , s'_{2}$''.
The algorithm grows from $s_1$ and trying to reach $s_2$.
     During the dynamic growing algorithm, when
 procedure ``touches'' any node $\overline{s}$ (when
 trying to grow to a new neighbor with the expanded cut) 
   that is in the
same face as $s_2$, then we need to consider
the two cuts. 

The first cut is  
between set   ``the current cut region from $s_1$ to $\overline{s}$  
union with all nodes  in the same face as $s_2$ grows clockwise from $\overline{s}$ to $s_2$''
and $s'_1,s'_{2}$, while
$s'_1$ and $s'_2$ are connectivity preserving,

The second cut is  
between set   ``the current cut region from $s_1$ to $\overline{s}$  
union with all nodes  in the same face as $s_2$ grows counterclockwise from $\overline{s}$ to $s_2$''
and $s'_1,s'_{2}$, while
$s'_1$ and $s'_2$ are connectivity preserving,
  We choose the smaller cut of these two as $min(\overline{s})$.

So if $min(\overline{s})$ is the smallest cut in all possible growing
choices, then we select $\overline{s}$ as the next growing point
and choose the corresponding cut to reach $s_2$.
Otherwise choose the other growing node that does not
touch the same face of $s_2$.

Note that if the dynamic growing algorithm
 ``touches'' multiple nodes in the same face  of $s_2$,
then all these nodes must be connected (otherwise there will be
hole), and we just consider these nodes as a single node and
continue with the same  procedure described above.

   Now we can see  if  during the growing process,
we never touch nodes in the same face of $s_2$
(except the last step, that we touch $s_2$), then $s'_1$ and $s'_2$
will always be connected during the procedure, and
we find the correct cut.

If during the   growing process,
we  touch nodes in the same face of $s_2$, then
we apply the above procedure. If this node $\overline{s}$ is
selected as the next growing node, then
the growing cut will include $s_2$, and $s'_1$ and $s'_2$ will
be connected, we find the correct cut again.
In this way, we can solve the 2-versus-2 case.

  If we directly apply the original procedure without
this modification, we may not be able to continue in
some intermediate step but the satisfying cut
still exists and we will miss it.
We have the following observations: (1) In every
step of the path growing procedure, the node $s_2$
may either be inside the component that contains
$s'_1$ and $s'_{2}$ defined by the new found CPMEC
 or inside the component that contains
$s_1$. (2) If at any step, the path cannot grow
and node $s_2$ still not been added, this means
it is infeasible to find a satisfiable cut. However,
this  infeasibility can be verified at the beginning
of the algorithm. So if there exists a feasible  cut,
this situation will not happen.
Since we can compute the
3-node planar CPMEC in polynomial time, the procedure
will be able to find the desired cut in polynomial time.
\end{proof}

Note that this result immediately implies that
the network diversion  problem in planar graphs is in $P$, since
the network diversion problem is a special case of
 2-node versus 2-node CPMEC.

\begin{corollary}
The network diversion problem in planar graphs can be solved in polynomial
time.
\end{corollary}

\begin{proof}
The network diversion problem in planar graphs is a special case
of the 2-node versus 2-node CPMEC in planar graphs where
node $s_1$ and $s'_1$ is connected.
\end{proof}

Another result is that the two node location constrained shortest path (LCSP)
problem~\cite{cpmc14} is in $P$. The original definition of LCSP
defines the shortest path where one node should be above the path. The
  two node LCSP problem is to  find the shortest path where one given
node is located above the path, another given node is located below the path.
\begin{corollary}
The two node  LCSP problem in planar graphs can be solved in polynomial
time.
\end{corollary}

\begin{proof}
The two node LCSP problem  is a special case of the 2-node versus 2-node CPMEC in planar graphs.
If we add a dummy node $s_1$ that
connects to all nodes that are in the boundary of the graph and above the
two end points of the path, and   a dummy node $s'_1$ that
connects to all nodes that are in the boundary of the graph and below the
two end points of the path, and consider the
two position defining nodes as $s_2$ and $s'_2$ respectively, then the
problem becomes a special case of the 2-node versus 2-node CPMEC in planar graphs
where $s_1$ and $s'_1$ are in the same surface.
\end{proof}

\section{Threshold Minimum Cut Problem}
\label{sec:thresh}

 We first formalize the  threshold minimum node cut problem as follows.

 Suppose we have a graph $G=(V,E)$ and a set of  $k$
service nodes $\Gamma=\{S_1,S_2,\ldots ,S_k\}$,
($S_i \in V$ for $1\leq i \leq k$),
a client node $A \in V$, and a threshold
 integer $l$. Every node
$v_i$ in $G$ has an associated cost
$c_i$, how to find a minimum
cost node cut such that at least
$l$ out of the $k$  service nodes
will be disconnected from $A$?

The decision version of the problem is:

  Suppose we have a graph $G=(V,E)$ and a set of  $k$
service nodes $\Gamma=\{S_1,S_2,\ldots ,S_k\}$,
($S_i \in V$ for $1\leq i \leq k$),
a client node $A \in V$,  a threshold
 integer $l$, and another integer
$B$. Every node
$v_i$ in $G$ has an associated cost
$c_i$. Can one find a
 node cut such that at least
$l$ out of the $k$  service nodes
will be disconnected from $A$ and the total
cost of the cut is no more than $B$?

 We name this problem as the threshold minimum
node cut problem (TMNC). The threshold minimum
edge cut problem (TMEC) can also be similarly
defined, the only difference is that
every edge in $G$ has a cost and  the
cut is an edge cut.

\subsection{Hardness of the  TMC problem}
The  TMNC problem was shown to be NP-complete in~\cite{DM13}.
 We can also show that the TMEC problem is NP-complete.

\begin{theorem}
The TMEC problem is NP-complete.
\end{theorem}

\begin{proof}
We can use the reduction from minimum bisection problem (with unit edge cost).
Given an instance of minimum bisection problem
$G=(V,E)$, we can construct an instance
of the threshold minimum
edge cut problem. We  construct a new graph $G'$ which contains all nodes in
$G$ with an additional node $A$. Node $A$ is connected with
every node in $G$ with an edge of cost $n^2$, where $n=|V|$.
Without loss of generality, we assume $n$ is
an even number.
We also define the set $\Gamma$ of the TMEC instance to be $V$, and $l=n/2$.
We have the following observations:

\begin{itemize}
\item{The threshold
minimum cut
in $G'$ will make $A$ to be separated from
at least $n/2$ nodes in $V$, that is, the threshold
minimum cut will have a value that is at least $n^2n/2=n^3/2$. }
\item{The threshold minimum  cut
in $G'$ will not make $A$ to be separated from more than
$n/2$ nodes in $V$ since every additional edge
in $G'$ has cost $n^2$, which is more than the number of
total edges in $G$.}
\end{itemize}

 Now we can see that the minimum cut in the threshold minimum
edge cut instance instance should include a minimum
bisection of $G$ and exactly $n/2$ edges between
$A$ and  nodes in $V$. If we can find
a threshold cut in $G'$ that is $u+n^3/2$, we can
find the corresponding bisection of $G$ with value $u$.
Conversely, if there exists a  bisection of $G$ with value $u$,
we can find threshold cut in $G'$ with value $u+n^3/2$.
 This finishes the reduction.

\end{proof}

 Based on  the NP-completeness
of the threshold minimum node cut problem, it is easy to see the
following problem is NP-complete~\cite{DM13}.

\begin{definition}
Given an undirected  graph $G=(V,E)$ and
a number $m$, can one find a subgraph
with at least $m$ edges, and the
number of nodes in the subgraph is minimized?
\end{definition}

Note that this is the inverse problem
of the maximum k-subgraph problem (unit weight
case). And this inverse k-subgraph problem
can be further generalized to
the following set minimum cover problem~\cite{DM13}:

\begin{definition}
Given a set $S$ of $n$ elements,
a collection $C$ of $m_1$ subsets of $S$, and positive integer
 $m \leq m_1$,
can one find $m$ subsets from $C$  such that
the total number of distinct elements in the union
of the $m$ subsets
is minimized?
\end{definition}

 We can see the set minimum cover problem is the
generalization of the inverse k-subgraph problem,
so  the set minimum cover problem is also
NP-complete.

Similarly we can define the  inverse problem of the  set minimum cover
 problem as follows~\cite{DM13}.
\begin{definition}
Given a set $S$ of $n$ elements,
a collection $C$ of $m_1$ subsets of $S$, and a positive integer
 $n_1 \leq n$,
can one find a subset $S'$ from $S$  such that
the size of  $S'$  is $n$
and the number of subsets (from $C$) fully covered by $S'$
is maximized?
\end{definition}

Here ``fully covered'' means that every element of the subset
is included in $S'$.
We denote this problem as the  set maximum cover problem. We can see
this problem is the generalization of the
 maximum k-subgraph problem.

If every set in the set minimum cover problem
has at most $\tau$ elements, we denote this special
case as the $\tau$-minimum cover problem.
Similarly we can define the $\tau$-maximum cover problem.

For the max k-subgraph problem, the best known approximation algorithm is the
 $n^{1/4}$ ratio approximation
algorithm in~\cite{BCCFV10}.
It is conjectured that the problem
cannot be approximated within $n^{\delta }$ for some
$0<\delta <1$, but currently the best known
hardness result is that it has no
polynomial time approximation scheme unless
 $NP \nsubseteq \cap_{\delta>0} BPTIME(2^{n^{\delta}})$~\cite{FKP99,Khot05}.

The relationship of the
approximability between the set minimum cover problem
and the set maximum cover problem is shown in the next theorem,
proved in~\cite{DM13} (the proof in the paper 
applies to 2-minimum cover problem, which can be easily
extended to the $\tau$-minimum cover problem). 

\begin{theorem}
If the $\tau$-minimum cover problem can be
approximated within ratio $1+\epsilon$ ($0<\epsilon <1$ ,that is,
there is an algorithm that can return
a subset of $S$ with no more than $1+\epsilon$ times number
of elements compared with the optimal solution),
  then the
$\tau$-maximum cover problem can be approximated within
 ratio $4^{\tau\epsilon }$ (that is,
suppose the optimal solution
of the $\tau$-maximum cover problem
is $OPT$, there exists an algorithm
that can return a subset of $S$ that fully covers
at least $4^{-\tau\epsilon } OPT$  number of subsets of $C$.
\end{theorem}

Note that  $4^{-\tau\epsilon }$ is very close to 1 if $\epsilon $ is
very close to 0.

\begin{corollary}
The the set minimum cover problem problem does not
have polynomial time approximation scheme unless $NP \nsubseteq \cap_{\delta>0} BPTIME(2^{n^{\delta}})$.
\end{corollary}

\begin{proof}
This follows from the above theorem and the
result that max k-subgraph problem
(special case of the $\tau$-maximum cover problem when
$\tau=2$ ) does not
have  polynomial time approximation scheme unless $NP \nsubseteq \cap_{\delta>0} BPTIME(2^{n^{\delta}})$.
\end{proof}

\begin{corollary}
Both the threshold minimum node cut problem
and the set minimum cover problem do not
have polynomial time approximation scheme unless $NP \nsubseteq \cap_{\delta>0} BPTIME(2^{n^{\delta}})$.
\end{corollary}

We can further show that the maximum cover problem
cannot be approximated within any constant ratio unless $NP \nsubseteq \cap_{\delta>0} BPTIME(2^{n^{\delta}})$.

\begin{theorem}
The maximum cover problem
cannot be approximated within any constant ratio unless $NP \nsubseteq \cap_{\delta>0} BPTIME(2^{n^{\delta}})$.
\end{theorem}

\begin{proof}
First we show that if $2\tau$-maximum cover can be approximated within ratio
$\zeta$ ($\zeta > 1$, then  $\tau$-maximum cover can be approximated within ratio
$\sqrt{\zeta}$. For a $\tau$-maximum cover instance with
set $S$, collection $C$, and integer $n_1$, we can
construct a $2\tau$-maximum cover instance as follows:
the set $S$ and integer $n_1$ is the same as that in
the $\tau$-maximum cover instance. We define
a new collection $C^2$ which contains the union of all possible
pairs of subsets in $C$, with all repetitions of subsets being kept. For example,
if  $C=\{ \{a_1, a_2\},\{a_2, a_3\} \}$, then
\[ C^2=\{ \{a_1, a_2\}, \{a_2, a_3\} ,  \{a_1, a_2, a_3\} ,  \{a_1, a_2, a_3\} \}   \]
We define $C^2$ to be the subset collection of the $2\tau$-maximum cover instance.
Suppose we have a ratio $\zeta$ approximation algorithm for
the $2\tau$-maximum cover instance, we can apply the the same solution (set of elements, denoted
as $S'$)
for the $\tau$-maximum cover instance. We have the following observations:
\begin{itemize}
\item{If $S'$ fully covers $\eta$ subsets of $C$, then it fully covers $\eta^2$ subsets of $C^2$. }
\item{If $S'$ fully covers $\eta$ subsets of $C^2$, then it fully covers $\sqrt{\eta}$ subsets of $C$. }
\item{If the optimal solution of the $2\tau$-maximum cover instance
can fully cover $OPT$ subsets of  $C^2$, then the optimal
solution of the $\tau$-maximum cover instance fully covers $\sqrt{OPT}$ subsets of $C$. }
\end{itemize}

Actually we can further show that any  subset of $S$ will fully cover a perfect square
number of subset in $C^2$. Suppose $S'$ fully covers $\eta$ subsets
in $C$, then $S'$ will cover all those subsets that are union of two
subsets in $C$, and $S$ will not fully cover any other subsets in $C^2$
because at least one of the elements of these subsets will not be in $S'$.
This means the number of fully covered subsets in $C^2$ by $S'$ will always be a perfect square.

So if we have a ratio $\zeta$ approximation algorithm for the $2\tau$-maximum cover problem,
then we will have a ratio $\sqrt{\zeta}$  approximation algorithm for the $\tau$-maximum cover problem.

If there exists a constant ratio (denoted as $\zeta$) approximation algorithm
 for the set maximum cover problem, we have the following
observation:
given any instance of the set maximum cover problem with maximum number of elements in the
collection of subsets to be $\tau$ and any number $\zeta_1 > 1$, we can
construct a $2^\mu\tau$-maximum cover instance where $\mu=\lceil \log(\zeta/\zeta_1) \rceil$,
with the same $S$ and $n_1$, but the collection is defined to be $C^{2^{\mu} }$. Here
   $C^{2^{\mu} }$ is the resulting collection of subsets by applying the
$C^2$ operation $\mu$ times iteratively. For example, $C^4$ is the
resulting collection which contains the union of all possible
pairs of subsets in $C^2$. Since the $2^\mu\tau$-maximum cover instance can be approximated
within $\zeta$, then we can use its result to obtain the $ \sqrt[2^{\mu}]{\zeta} \leq \zeta_1$
approximation for the original $\tau$-maximum cover instance. This means
if  there exists a constant ratio approximation algorithm
 for the set maximum cover problem, then we can find a
polynomial time approximation scheme for it. However we know that
maximum $k$-subgraph problem (which is equivalent to the  2-maximum cover problem)
does not have polynomial time approximation scheme unless $NP \nsubseteq \cap_{\delta>0} BPTIME(2^{n^{\delta}})$, so there
exists no constant ratio approximation for the set maximum cover problem unless
$NP \nsubseteq \cap_{\delta>0} BPTIME(2^{n^{\delta}})$.
\end{proof}

\begin{corollary}
The network inhibition problem (the goal is to
find the most effective way to reduce the capacity
of a network flow within fixed budget)
and the network interdiction
problem (a different version of network inhibition,
and the goal is to  choose a subset of arcs to delete, without
exceeding the budget, that minimizes the maximum
flow or other flow metrics
that can be routed through the network induced on the remaining arcs)
  cannot be approximated within any constant ratio unless
$NP \nsubseteq \cap_{\delta>0} BPTIME(2^{n^{\delta}})$.
\end{corollary}

\begin{proof}
We can reduce the set maximum cover problem to the network inhibition
problem and network interdiction
problem. For an instance of the set maximum cover problem
with set $S$, collection $C$ and integer $n_1$, we can construct
a directed graph as follows. We first create a source node $U$  and destination
node $T$. For every subset in $C$, we create
a node. For every element in $S$, we create an arc, which is again
connecting to every node whose corresponding subset contains
the element. The source is connected to the starting point of every
arc that corresponds to every element. Every node
that corresponds to a subset is connected to the destination.
The blocking cost of the arcs that corresponds to  the elements is
1, and the cost of all other arcs is infinity.
We also set the capacity of the  incoming arcs
of destination to be 1, all other arcs have capacity infinity.  As an example,
suppose $S=\{a_1,a_2,a_3\}$, $C=\{ \{a_1,a_2\}, \{a_2,a_3\}\}$,
then the constructed network inhibition/interdiction instance is shown in
Fig.~\ref{fig:inh}. In the figure, the first number corresponds to every arc  is the
capacity of the arc, and the second number corresponds to every arc is the
blocking cost of the arc.

\begin{figure}
\centering

\includegraphics[width=2.5in]{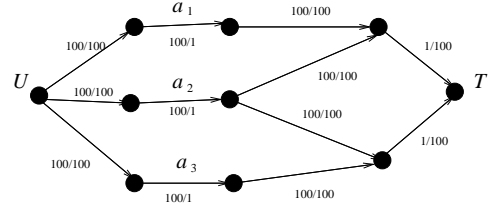}
\caption{The constructed  network inhibition/interdiction instance}
\label{fig:inh}
\end{figure}

Now it is easy to see the maximum flow
from $U$ to $T$ is $|C|$. If a subset of $S$ fully covers
any subset in $C$, the deletion of  the corresponding arcs
will reduce the maximum flow by 1. We can see that
the  set maximum cover in the original instance
corresponds to a partial cut for the flow. In this way we reduce the
 set maximum cover problem to the
 network inhibition
problem and network interdiction
problem.
 So the network inhibition problem and the network interdiction
problem cannot be approximated within any constant ratio unless
$NP \nsubseteq \cap_{\delta>0} BPTIME(2^{n^{\delta}})$.
\end{proof}

\section{Approximation Algorithms of TMC}
\label{sec:alg}

First we present a ratio $O(\sqrt{n})$ approximation algorithm
for TMNC problem in Algorithm~\ref{alg:node}.
In the algorithm, we assume that $l \geq \sqrt{n}$, since
it is trivial to have an $\sqrt{n}$  approximation
when $l < \sqrt{n}$ (one just needs to sort
$S_1, \ldots, S_k$ according to their minimum cut values with $A$
in ascending order, choose the first $l$ nodes and  find the minimum
cut  between these nodes and $A$).

\begin{algorithm}
  1. Solve the following Linear Programming (LP):

  \[ \text{  Maximize  } \displaystyle\sum_{\text{ all nodes } v_i}
        X_ic_i  \]

  Subject to

    \[ Y_i \leq X_i + Y_j\, ,\text{ for all
       neighbors }v_j \text{ of } v_i\, ,\forall v_i  \]
   \[ 0 \leq X_i \leq 1\,, 0 \leq Y_i \leq 1\,,\forall v_i \]
   \[ Y_A = 0\]
   \[ \displaystyle\sum_{i=1}^kY_{S_i} \geq l\]

 After solving the LP, sorting the nodes $S_1,S_2, \ldots S_k$
 according to the LP value $Y_{S_i}$
in descending order.

 2.  Find the first $S_i$ in the sorted list  $S_1, \ldots, S_k$ that is less than
 $1/\sqrt{n}$.

  3. \If {$i > l$} {Find the minimum cut between $S_1, \ldots, S_{l}$ and $A$, return
 this cut value.}
    \Else {
Sorting all $S_j$ ($j>i$ ) in ascending order
according to the cut value $c(S_j)$, here $c(S_j)$ is the minimum cut value
between node $S_j$ and node $A$. Denote the first $l-i+1$ nodes in
the sorted list as $S'_1, \ldots, S'_{l-i+1}$.
Find the minimum cut between $S_1, \ldots, S_{i-1}, S'_1, \ldots, S'_{l-i+1} $ and $A$.
return this cut value. }
 \caption{The Approximation Algorithm for Threshold Minimum Node cut}
\label{alg:node}
\end{algorithm}

\begin{theorem}
Algorithm~\ref{alg:node} achieves a ratio of $O(\sqrt{n})$.
\end{theorem}

\begin{proof}
The LP defines a fractional cut between the nodes
$S_1,S_2, \ldots S_k$ and $A$, and the summation
of the accumulated cut value of nodes $S_1,S_2, \ldots S_k$ is
at least $l$. If we sort the nodes $S_1,S_2, \ldots S_k$
 according to their cut value ($Y_{S_i}$)
in descending order and there are at least $l$ nodes in the list
than has cut value at least $1/\sqrt{n}$, then the minimum cut
between these $l$ nodes and $A$ will be at most $\sqrt{n}$ times the
fractional cut value returned by the LP in the algorithm (denoted as $c(LP)$).
But $ c(LP) \leq OPT$, where $OPT$ is the minimum threshold node cut value.
So in this case the algorithm achieves the ratio $\sqrt{n}$.
If there less than $l$ nodes in the sorted list that is less than
 $1/\sqrt{n}$, in the algorithm the minimum cut between   $S_1, \ldots, S_{i-1}$
and $A$ is less than   $\sqrt{n}OPT$.
The minimum cut between   $S'_1, \ldots, S'_{l-i+1}$
and $A$ is also less than   $\sqrt{n}c(S'_{l-i+1}) \leq \sqrt{n}OPT $,
because there are at most $\sqrt{n}$ nodes (in $S_1, \ldots, S_k$)
 that have a cut value less than $1/\sqrt{n}$, otherwise the total cut value
will be less than $l$ in the LP formulation.
So in this case the solution returned by the algorithm
will also be at most  $2\sqrt{n}OPT$.

\end{proof}

Next we present a ratio $O(\log^2{n})$ algorithm for the TMEC
 problem in Algorithm~\ref{alg:edge}.

\begin{algorithm}
1. Generate $k-1$ cliques $\gamma_1, \ldots, \gamma_{k-1}$
with size $n^2$ and one clique $\gamma_k$ with size   $(k-1)n^2$.
The cost of all edges in the cliques are set to be $n^2$.

2. \For{$i=1..k$} {  Connect node $S_i$ to an arbitrary node in $\gamma_k$ with an edge
of cost $n^2$, and connect each of the  remaining $S_j$ ($j\neq i$) to
one clique $\gamma_u$ ($ 1 \leq u \leq k-1$), with an edge of cost $n^2$.

 \For{$j= ((2l-2)n^2-n+l).. (2(k-1)n^2+n-2 )$ } {
    Generate a clique $\gamma'$ with size $j$. The cost of all edges in $\gamma'$ are set to be $n^2$.
    Connect node $A$ to an arbitrary node in $\gamma'$ with an edge
of cost $n^2$. Denote the current graph (with $l+1$ cliques added) as $G'$. Apply
    the $O(\log^2{n})$ approximation algorithm on $G'$ to find the minimum bisection.
 Denote the value of the bisection as $B(i,j)$.
      }
}
3. Find the minimum value of all $B(i,j)$, and return the  bisection
corresponds to this $B(i,j)$ as the threshold edge cut.
\caption{The Approximation Algorithm for Threshold Minimum Edge cut}
\label{alg:edge}
\end{algorithm}

We can prove that this algorithm can achieve an approximation ratio $O(\log^2{n})$.

\begin{proof}
We can consider the minimum threshold edge cut. Suppose the
 the minimum threshold edge cut separate $S_{\alpha_1}, \ldots ,  S_{\alpha_w}$
from $A$ and $w \geq l$. In the minimum threshold edge cut, the maximum number of nodes ( in $G$ )
that can be separated with $A$ can be as large as $n-1$ ($A$ is separated from all other
nodes),  and as small as $l$. We denote this number
as $\beta$. First we can see that no new added
edges in $G'$ will be in the minimum bisection since
all the new edges have a cost $n^2$, which is larger than
the total number of edges in $G$.
When $S_i \in \{S_{\alpha_1}, \ldots ,  S_{\alpha_w}\}$,
 one of the minimum bisection (in $G'$) with value $B(i,j)$ ( $j$
ranges from $(2l-2)n^2-n+l$ to  $2(k-1)n^2+n-2$ ) is equivalent
to the minimum threshold edge cut. When $\beta$ is $n-1$,
 the minimum bisection $B(i,j)$ in $G'$ with $j=  2(k-1)n^2+n-2$
 corresponds
to the minimum threshold edge cut in $G$, since in this case
all the other $n-1$ nodes in $G$ and all the cliques
appended to $S_1, \ldots, S_k$ ( total size $2(k-1)n^2$)
will be separated from $A$, but node $A$ and the appended
clique of $A$ will have total size $ 2(k-1)n^2+n-2 + 1 = 2(k-1)n^2 + n -1 $.
When $\beta$ is $l$,
 the minimum bisection $B(i,j)$ in $G'$ with $j= (2l-1)n^2-n+l$
 corresponds
to the minimum threshold edge cut in $G$, since in this case
the $l$ nodes in $S_{\alpha_1}, \ldots ,  S_{\alpha_l}$ and all the cliques
appended to then ( total size $(l-1+k-1)n^2$)
will be separated from $A$, but there are  $ n-l + (k-l)n^2 + (2l-2)n^2-n+l  = (l+k-2)n^2 $
remaining nodes in $G'$. So the minimum bisection
in $G'$ corresponds to a minimum threshold edge cut in $G$.
For all $\beta$ values between $l$ and $n-1$,
the algorithm will also find the corresponding minimum bisection
with appropriate $j$.  Since minimum bisection can be solved with
approximation ratio $O(\log^2{n})$, the above algorithm
also finds the minimum threshold edge cut with
ratio  $O(\log^2{n})$.

\end{proof}

\section{ Future Work}
\label{sec-fw}

The hardness of the   3-node  CPMEC (in undirected graphs) is still open, but we have some
interesting observations for the problem:
\begin{itemize}
\item{It is intriguing that
the algorithm does not work for general  graphs. It would be interesting to
classify the types of graphs that the  algorithm can find the optimal cut. }
\item{ We already have multiple hardness results for 3-node
 CPMNC but  the hardness of 3-node undirected CPMEC is still open since the
hardness proof of CPMNC cannot be applied to CPMEC. For many
other minimum cut based problems, such as the basic minimum cut and the minimum
multi-terminal cut there is no big difference between the hardness
of node cut and edge cut. So we conjecture that  3-node undirected CPMEC is also hard to solve, though it is not clear whether an NP-hard proof is available.
Even if it is not NP-complete, it may   not be in $P$, based on the assumption
that $P \neq NP$.}
\item{There are several minimum cut related problems which are NP-hard
in general case but have polynomial algorithms in planar graphs. The max-cut
problem and the minimum multi-terminal cut problem have polynomial time algorithm
in planar graphs~\cite{Dahlhaus94}. The hardness of Steiner tree problem in planar graphs
is still open but it has polynomial
time approximation scheme~\cite{BKK07}. The minimum multi-way
cut in planar graphs is NP-hard but has  polynomial
time approximation scheme~\cite{BHKM12}. The hardness of minimum
bisection in planar graphs is still open~\cite{Karpinski02}. It is important to investigate  what kind
of minimum cut related problems in planar graphs can be solved in polynomial
time and what is the deep logic behind this. Further research on this can provide
guidance on new problems related to planar minimum cut. }
\item{There is another similarity between minimum multi-terminal cut problem
and the CPMC in planar graphs.
Actually if
we  adopt the perturbation method, we can have a simple algorithm for minimum 3-terminal cut
in planar graphs.
The CPMC problem can be considered as the
``complementary'' problem of the minimum multi-terminal  cut problem.
Further investigation of the relationship between the two problems
will help the understanding of both problems.
}
\item{It is rather surprising that the dynamic programming  algorithm works for
3-node planar CPMEC. It is interesting to further investigate what kind of
constrained minimum cut problem can be solved by similar  algorithms. }
\end{itemize}

We conjecture that the 3-node undirected CPMEC problem may belong to a class
of problems that are neither in $P$ nor $NP$-complete. Thus, the problem
may be related to the central question of $NP$ versus $P$.
For some special and practical graphs,
we believe that there may exist efficient
precise or approximation algorithms, which will
be another future research direction.
For TMC problem, there is much room to improve the
the approximation ratio and hardness result.
%\bibliography{cpmc}

\end{document}